\documentclass[11pt]{article}
\usepackage[utf8]{inputenc}
\usepackage[top = 1in, bottom = 1in, left = 1in, right =1in]{geometry}
\usepackage{latexsym,amsmath,amssymb,amsthm, graphicx}
\usepackage{algorithm,algorithmicx,algpseudocode}
\usepackage{enumitem}
\usepackage{hyperref}
\usepackage{mathtools}
\usepackage{xcolor}
\usepackage{caption}
\usepackage{subcaption}
\usepackage{graphicx}
\usepackage{soul}
\usepackage{bbold}
\usepackage{dsfont}
\usepackage{appendix}
\usepackage{cite}
\usepackage[makeroom]{cancel}
\allowdisplaybreaks
\makeatletter
\newcommand{\SB}[1]{#1^+} 
\newcommand{\pushright}[1]{\ifmeasuring@#1\else\omit\hfill$\displaystyle#1$\fi\ignorespaces}
\newcommand{\pushleft}[1]{\ifmeasuring@#1\else\omit$\displaystyle#1$\hfill\fi\ignorespaces}
\makeatother

\setlist[itemize]{leftmargin=*}

\newcounter{thm_count}
\newcounter{claim_count}
\theoremstyle{remark} 
\newtheorem*{theorem*}{Theorem} 
\newtheorem{theorem}[thm_count]{\bf Theorem} 
\newtheorem{lemma}[thm_count]{\bf Lemma} 
\newtheorem{claim}[claim_count]{\bf Claim} 

\newcounter{assump}

\theoremstyle{remark}

\newtheorem{assumption}[assump]{\bf Assumption}

\renewenvironment{proof}[1][Proof:]{\begin{trivlist}
\item[\hskip \labelsep {\bfseries #1}]}{\end{trivlist}}

\title{Distributed Optimization of Convex Sum of Non-Convex Functions \footnote{This research is supported in part by National Science Foundation. Any opinions, findings, and conclusions or recommendations expressed here are those of the authors and do not necessarily reflect the views of the funding agencies or the U.S. government.}}

\author{Shripad Gade $\qquad$ Nitin H. Vaidya \\ \\ Department of Electrical and Computer Engineering, and \\ Coordinated Science Laboratory, \\ University of Illinois at Urbana-Champaign. \\ Email: \{gade3, nhv\}@illinois.edu \\ \\ Technical Report}

\date{\today}

\begin{document}

\maketitle

\begin{abstract}
We present a distributed solution to optimizing a convex function composed of several non-convex functions. Each non-convex function is privately stored with an agent while the agents communicate with neighbors to form a network. We show that coupled consensus and projected gradient descent algorithm proposed in \cite{ram2010distributed} can optimize convex sum of non-convex functions under an additional assumption on gradient Lipschitzness. We further discuss the applications of this analysis in improving privacy in distributed optimization.  
\end{abstract}

\section{Introduction}
Distributed convex optimization has found numerous applications in resource allocation \cite{xiao2006optimal}, robotics \cite{bullo2009distributed}, machine learning \cite{boyd2011distributed} etc. Several works have appeared in distributed convex optimization domain over the past decade \cite{nesterov2012efficiency,liu2015asynchronous,agarwal2011distributed,Singh2014,Nedic2007,nedic2011asynchronous,Nedi2015,zhang2014asynchronous,huang2015differentially,rabbat2004distributed,zhu2012distributed,gade16distoptclientserver,sra2012optimization}. Recently, a few papers \cite{bianchi2013convergence,di2016next,sun2016distributed} have dealt with distributed non-convex optimization \footnote{For a thorough treatment of prior literature, readers are directed to papers cited here and references therein.}. The classical distributed optimization problem involves finding a state vector $x^*$ in a feasible set $\mathcal{X}$ such that $x^* \in {\text{argmin }}_{x \in \mathcal{X}} \sum_{i=1}^S f_i(x)$, assuming each individual objective function $f_i(x)$ being convex and its gradients being bounded. 

In this report we prove that the convexity assumption on individual objective functions can be relaxed as long as the sum is still a convex function and individual function gradients are Lipschitz continuous. Formally, we can find $x^*$ using a distributed protocol (Algorithm~\ref{Algo:IterDistOptNCFun}) such that, $$x^* \in \underset{x \in \mathcal{X}}{\text{argmin }} f(x)  \triangleq \sum_{i=1}^S f_i(x),$$ where the individual functions $f_i(x)$ may be non-convex, however $f(x)$ is convex. We designate this partitioning setting ($f_i(x)$ being non-convex while $f(x)$ being convex) as convex sum of non-convex functions.

\subsection{Motivating Example - Privacy Enhancing Distributed Optimization} \label{Sec:Motivation}

Let us review a simple distributed optimization problem to demonstrate how privacy requirements motivate this work. We consider $S = 3$ agents each endowed with a private, convex function $f_i(x)$, where agents intend to minimize the convex function $f(x) = \sum_{i = 1}^S f(x)$, while having access to only their own private objective function. This is a standard distributed convex optimization problem that has been extensively studied over the past decade. Figure~\ref{Fig:F-1} denotes the communication topology of the agents. Distributed algorithms (e.g. Algorithm~\ref{Algo:IterDistOptNCFun}, \cite{ram2010distributed}) solve this problem by sharing states and performing local gradient descent updates (assuming some underlying connectivity among agents). Now consider a scenario where agents wish to introduce privacy in the optimization protocol. 

\begin{figure}[h]
\centering
  \includegraphics[width=.5\linewidth]{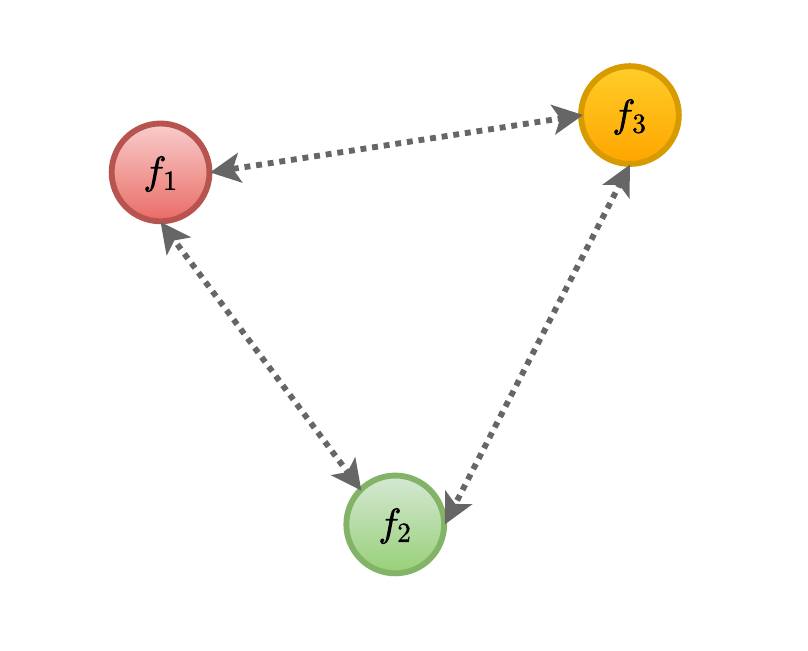}
  \caption{Standard distributed optimization problem: $f_i(x)$ are private, convex functions, $\|\nabla f_i(x)\| \leq L_i$ and we intend to minimize $f(x) = \sum_{i = 1}^S f(x)$ over some constraint set $\mathcal{X}$.}
  \label{Fig:F-1}
\end{figure}

Let every agent partition its objective function into additive components. That is, $f_1(x) = f_{1,2}(x) + f_{1,3}(x)$, $f_2(x) = f_{2,1}(x) + f_{2,3}(x)$ and $f_3(x) = f_{3,1}(x) + f_{3,2}(x)$ where the components may or may not be convex (see Figure~\ref{Fig:F-2}). The protocol is now modified so that in the gradient descent step (Eq.~\ref{Eq:ProjGradALG2}), agents use different component gradients instead of using $f_i(x)$. For example, when Agent 2 needs to perform gradient descent step and share states with Agent 1 and 3, Agent 2 uses $\nabla f_{2,1}(x)$ for descent update and sends state to Agent 1; and uses $\nabla f_{2,3}(x)$ for descent and sends states to Agent 3. Similarly, Agent 3 sends state to Agent 1 after performing gradient descent using $\nabla f_{3,1}(x)$ and uses $\nabla f_{3,2}(x)$ for descent before sending updates to Agent 2. Application of different gradients allows the transmitted states (sent to different agents) to be different. Agents can hide the information about their individual (and private) objective functions $f_i(x)$. This improves privacy while still allowing all agents to learn the correct model (reach optimum). 

\begin{figure}[h]
\begin{subfigure}{0.48\textwidth}
  \centering
  \includegraphics[width=.9\linewidth]{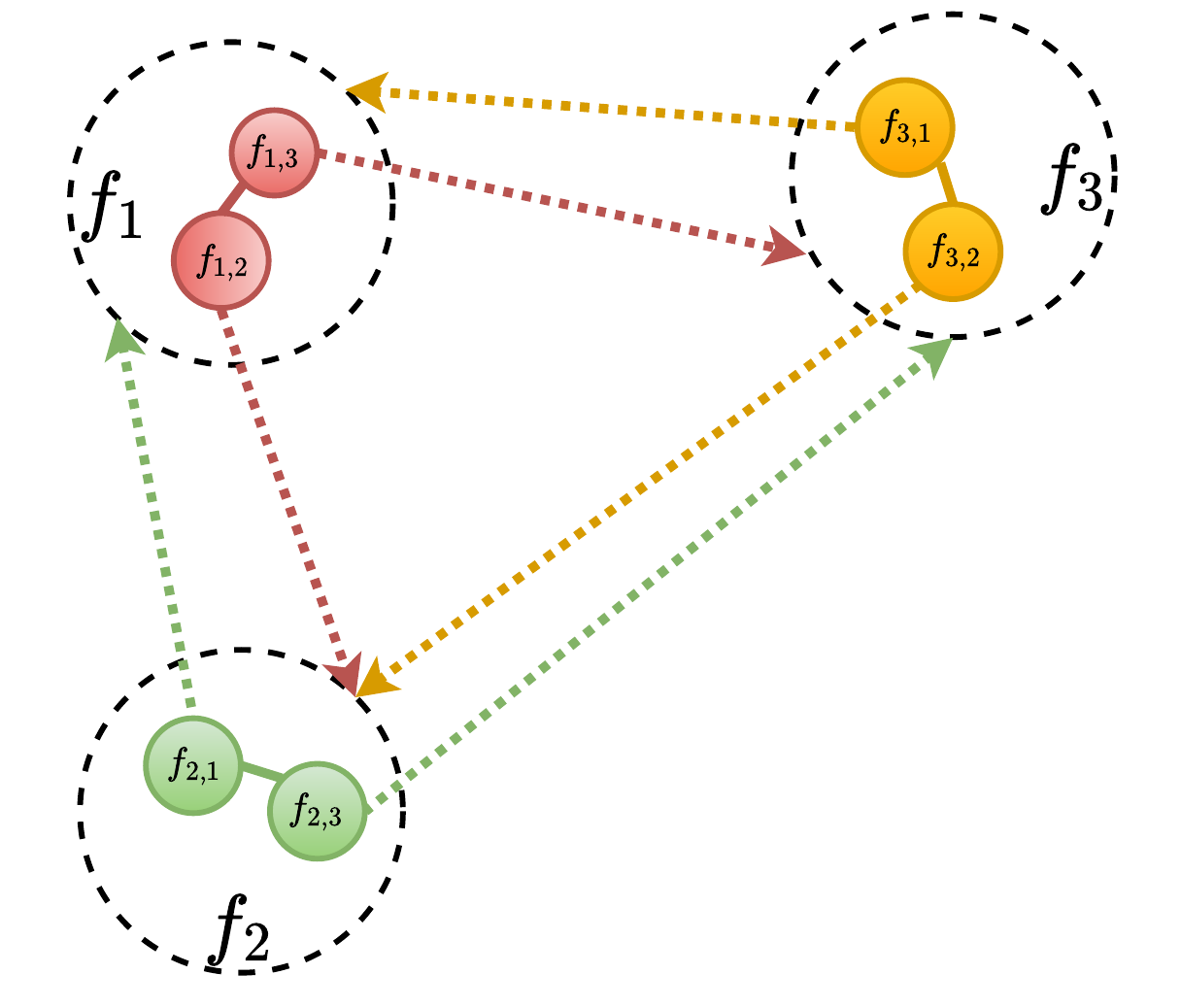}
  \caption{Individual functions are partitioned into possibly non-convex components for privacy.}
  \label{Fig:F-2}
\end{subfigure} \hfill
\begin{subfigure}{0.48\textwidth}
  \centering
  \includegraphics[width=.95\linewidth]{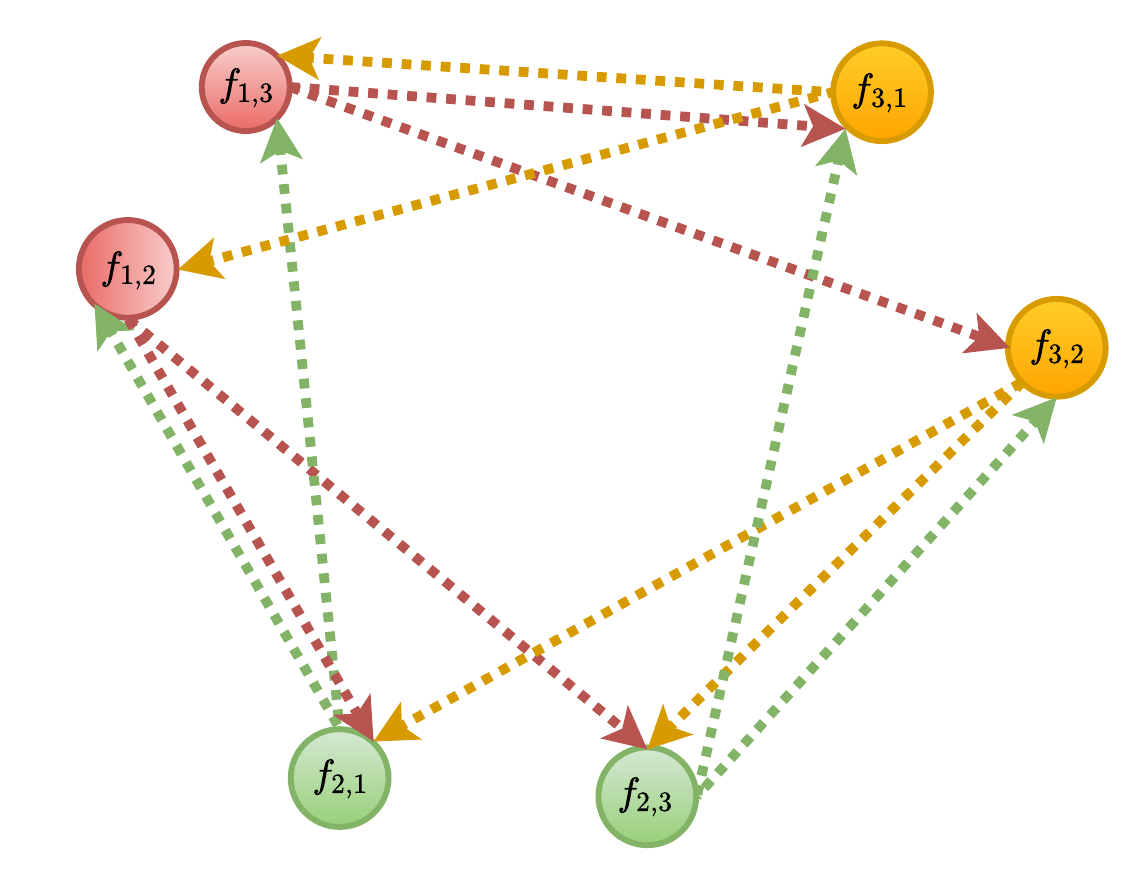}
  \caption{Distributed optimization of convex sum of non-convex functions.}
  \label{Fig:F-3}
\end{subfigure}%
\caption{A Motivation Example: Privacy Enhancing Distributed Optimization}
\label{Fig:MotExample}
\end{figure}

We now view each component function as being associated to a virtual agent, with any agent being represented by two virtual agents in our example ($S = 6$). The links between virtual agents must be a subset of the links between the corresponding real agents. By choosing an appropriate set of links
between the virtual agents, we can ensure that the network of virtual agents is strongly connected, provided that the original network is strongly connected. Figures~\ref{Fig:F-2} and \ref{Fig:F-3} present an example.This is motivated by the desire to improve privacy. The topology in Figure~\ref{Fig:F-2} should enhance privacy in the sense that a single real agent (on its own) will
not learn the cost function of other agents.

Note, that the individual components (in our example) are not necessarily convex; and we intend to minimize their sum, $$f(x) = f_{1,2}(x) + f_{1,3}(x) + f_{2,1}(x) + f_{2,3}(x) + f_{3,1}(x) + f_{3,2}(x).$$ This new problem fits exactly in the structure defined in Section~\ref{Sec:Problem}. As long as the original network in Figure~\ref{Fig:F-2} is strongly connected, we can ensure that the new network in Figure~\ref{Fig:F-3} also stays strongly connected. The doubly stochastic, transition matrix $B_k$ can be defined as,
\begin{align}
    B_k = 
    \begin{bmatrix}
    (1-2\kappa)& 0& \kappa& 0 & 0 & \kappa\\
    0 & (1-2\kappa) & \kappa& 0 & 0 & \kappa\\
    0 & \kappa& (1-2\kappa) & 0 & \kappa& 0 \\
    0 & \kappa& 0  & (1-2\kappa) & \kappa& 0 \\
    \kappa& 0 & 0 & \kappa& (1-2\kappa) & 0  \\
    \kappa& 0 & 0 & \kappa& 0 & (1-2\kappa) \\
    \end{bmatrix}. && \ldots \kappa\geq 0
\end{align}

None of the agents sees information from all virtual agents (part of any specific agent). Thus, agents can protect their private objective function yet solve the distributed optimization problem correctly. In the report we provide theoretical analysis and convergence proofs for distributed optimization algorithm in \cite{ram2010distributed} to optimize convex sum of non-convex functions $f(x)$, under any arbitrary time varying connected topology. We also present in Section~\ref{Sec:Privacy}, a different privacy enhancing scheme based on secure multi-party aggregation strategy proposed in \cite{abbe2012privacy}. We discuss the application of analysis developed in this report for proving correctness and characterizing convergence of this new scheme. 

\subsection{Notation}
The number of agents is denoted by $S$. Upper case alphabets ($I, J, K$ etc.) are used to index agents. We use the symbol ``$\sim$" to denote communication link and information sharing between agents. As an example, $I \sim G$ denotes that agents $I$ and $G$ have a communication link between them, and conversely $I \cancel{\sim} J$ denotes that agents $I$ and $J$ can not share information (cannot communicate) with each other. The neighborhood set of agent $J$ is denoted by $\mathcal{N}_J$. The dimension of the problem (number of parameters in the decision vector) is denoted by $D$.

Iterations number is denoted by $k$. The decision vector (also referred to as iterate from now on) stored in agent $I$ at time (iteration) $k$ is denoted by $x^I_{k}$, where the superscript denotes the agent-id, the subscript denotes the time index. $x^J_{k}[p]$ ($p = 1, 2, \ldots, D$) denotes the $p^{th}$ dimension in decision vector $x^J_{i,k}$. The average of iterates at time instant $k$ is denoted by $\bar{x}_{k}$. 
\begin{equation}
\bar{x}_{k} = \frac{1}{S} \sum_{J=1}^S x^J_{k}. \label{Eq:deltaDef}
\end{equation}
We denote the disagreement of an iterate ($x^J_{k}$) with the iterate average ($\bar{x}_{k}$) by $\delta^J_{k}$.
\begin{equation}
\delta^J_{k} = x^J_{k} - \bar{x}_{k}. \label{Eq:deltaDef2}
\end{equation}

\noindent We use $\tilde{.}$ to denote a vector that is stacked by its coordinates. As an example, consider three vectors in $\mathbb{R}^3$ given by ${a_1} = [a_x, \ a_y, \ a_z]^T$, ${a_2} = [b_x, \ b_y, \ b_z]^T$, ${a_3} = [c_x, \ c_y, \ c_z]^T$. Let us represent ${a} = [{a_1}, \ {a_2}, \ {a_3}]^T$, then $\tilde{{a}} = [a_x, \ b_x, \ c_x, \ a_y, \ b_y, \ c_y, \ a_z, \ b_z, \ c_z]^T$. Similarly we can write stacked model parameter vector as, 
\begin{equation}
\tilde{x}_{0,k} = [x^1_{0,k}[1], x^2_{0,k}[1], \ldots, x^S_{0,k}[1], x^1_{0,k}[2], x^2_{0,k}[2], \ldots, x^S_{0,k}[2], \ldots, x^1_{0,k}[D], \ldots, x^S_{0,k}[D]]^T. \label{Eq:TildeVecDef}
\end{equation}

We use $g_h(x^J_{k})$ to denote the gradient of function $f_h(x)$ evaluated at $x^J_{k}$. Let, $L_h$ be the bound on gradient (see Assumption~\ref{Asmp:SubBound}) and $N_h$ be the Lipschitz constant (see Assumption~\ref{Asmp:GradLip}), for all $h = 1, 2, \ldots, S$. We define constants $\SB{L} = \sum_{h=1}^S L_h$ and $\SB{N} = \sum_{h=1}^S N_h$ to be used later in the analysis. $\|.\|$ denotes standard Euclidean norm for vectors, and matrix 2-norm for matrices ($\|A\| = \sqrt{\lambda_{\max}(A^\dagger A)} = \sigma_{\max}(A)$, where $A^\dagger$ denotes conjugate transpose of matrix A, and $\lambda, \ \sigma$  are eigenvalues and singular values respectively). 

Throughout this report, we will use the following definitions and notation regarding the optimal solution ($x^*$), the set of all optima ($\mathcal{X}^*$) and the function value at optima ($f^*$),
$$ f^* = \inf_{x \in \mathcal{X}} f(x), \qquad \mathcal{X}^* = \{x \in \mathcal{X} | f(x) = f^*\}, \qquad dist(x, \mathcal{X}^*) = \inf_{x^* \in \mathcal{X}^*} \|x - x^*\|.$$ 
\noindent The optimal function value, at the solution of the optimization problem or the minimizing state vector is denoted by $x^*$, is denoted by $f^*$.

\subsection{Organization}
Problem formulation, assumptions and framework is presetned in Section~\ref{Sec:Problem}. Consensus and Projected Gradient based algorithm (similar to ~\cite{nedic2009distributed,ram2010distributed})  is summarized in Section~\ref{Sec:Algorithm}. Convergence analysis, and correctness proofs are presented in Section~\ref{Sec:ConvAnalysis}. Discussion on applications of this framework to privacy is presented in Section~\ref{Sec:Privacy}.

\section{Problem Formulation and Algorithm} \label{Sec:Problem}
Let us consider $S$ agents, each of whom has access to a private, possibly non-convex function $f_i(x)$. We intend to solve the following optimization problem in a distributed manner,
$$\text{Find} \; x^* \in \underset{x \in \mathcal{X}}{\text{argmin}} \; f(x), $$
where $f(x) = \sum_{i=1}^S f_i(x)$ is a convex function. The agents communicate with their neighbors and share parameter state estimates. This communication graph is assumed to be bidirectional and connected (see Assumption~\ref{Asmp:DeltaConn}).  

We enforce the following assumption on the functions $f_i(x)$ and on the decision set, $\mathcal{X}$.  
\begin{assumption} [Objective Functions] \label{Asmp:Function} 
The objective functions $f_i : \mathbb{R}^D \rightarrow \mathbb{R}, \; \forall \; i = 1, 2,{ }\ldots, S$ may or may not be convex functions of model parameter vector $x$. However, the sum of individual objective functions is convex, i.e., $f(x) := \sum_{i=1}^S f_i(x)$ is convex. 
\end{assumption}
\begin{assumption} [Decision Set]
The feasible parameter vector set, $\mathcal{X}$, is a non-empty, closed, convex, and compact subset of $\mathbb{R}^D$. \label{Asmp:Set}
\end{assumption}
\noindent We make a boundedness assumption on the gradient of function $f_i(x)$ in Assumption~\ref{Asmp:SubBound}. And make an additional assumption on the gradients $g_h$ of functions $f_h$.  
\begin{assumption} [Gradient Boundedness] \label{Asmp:SubBound}
Let $g_i(x)$ denote the gradient of the function $f_i(x)$. There exist scalars $L_1, L_2, \ldots, L_S $ such that, $\| g_i(x) \| \leq L_i; \; \forall \; i = 1, 2, \cdots, S \; \text{and} \; \forall x \in \mathcal{X}$.
\end{assumption}
\begin{assumption}[Gradient Lipschitzness]
\label{Asmp:GradLip}
Each function gradient ($g_h(x)$) is assumed to be Lipschitz continuous i.e. there exist scalars $N_h > 0$ such that, $\|g_h(x) - g_h(y) \| \leq N_h \| x -y\|$ for all $x \neq y$ ($x,y \in \mathcal{X}$) and for all $h = 1, 2, \ldots, S$.
\end{assumption}

Agents are connected in an arbitrary time-varying topology, albeit under the assumption that agents form a connected component, Assumption~\ref{Asmp:DeltaConn}. 
\begin{assumption} [Connectedness]
\label{Asmp:DeltaConn}
At every iteration, $k$, there exists a path between any two agents. (Agents form a connected component.)
\end{assumption}

For the purpose of this report we will assume without explicitly stating that all communication links are synchronous and loss less. All agents are assumed to operate perfectly and do not experience any faults.

\subsection{Algorithm} \label{Sec:Algorithm}
We consider iterative algorithm for distributed convex optimization presented in \cite{ram2010distributed}. We show that the existing algorithm can optimize convex sum of non-convex functions. The first step in the algorithm is to fuse information from the neighbors and build an estimate of average of the parameter vector. A doubly stochastic matrices $B_k$, with the property that any entry $B_k[I,J]$ is greater than zero if and only if $I$ and $J$ can communicate with each other, is used for information fusion. Also, we assume that all non-zero entries are lower bounded by $\eta$, i.e. if $B_k[I,J] > 0$ then $B_k[I,J] \geq \eta$ for some constant $\eta>0$. If $\mathcal{N}_J$ denotes the set of agents that can send information to agent $J$, then we can write the fusion step as,
\begin{align}
    v^J_{k} = \sum_{I \in \mathcal{N}_J} B_k[J,I] x^I_{k}. 
    \label{Eq:InfFusALG1}
\end{align}

The information aggregation step is followed by projected gradient descent step. The descent step is formally written as, 
\begin{align}
    x^J_{k+1} = \mathcal{P}_{\mathcal{X}}\left[v^J_{k} - \alpha_k g_J(v^J_{k}) \right]. \label{Eq:ProjGradALG2}
\end{align}
Projected gradient descent is a well known iterative gradient based method that guarantees convergence to optimum under reducing learning rate ($\alpha_k$) \cite{bertsekas1976goldstein}. We assume that the monotonically non-increasing
learning rate/step-size possesses the following properties,
\begin{align}
     \alpha_k > 0, \ \forall k \geq 0; \quad \alpha_{k+1} \leq \alpha_k, \ \forall k \geq 0; \quad  \sum_{k=0}^\infty \alpha_k = \infty; \ \text{and} \quad \sum_{k=0}^\infty \alpha_k^2 < \infty. \label{Eq:LearnStepCond}
\end{align}

\begin{algorithm}
\caption{Distributed Algorithm for Optimization of Convex sum of Non-Convex functions}
\begin{algorithmic}[1]
\State Input: $x^J_k$, $\alpha_k$, NSteps \Comment{NSteps - Termination Criteria}     
\State Result: $x^* = \underset{x \in \mathcal{X}}{\text{argmin}} \sum_{i=1}^{S} f_i(x) $ 
\For {k = 1 to NSteps} 
    \For {J = 1 to S}
        \State $v^J_{k} = \sum_{I \in \mathcal{N}_J} B_k[J,I] x^I_{k}$   \Comment{Information Fusion}
        \State $x^J_{k+1} = \mathcal{P}_\mathcal{X} \left[ v^J_{k} - \alpha_k g_J(v^J_{k})\right]$ \Comment{Projected Gradient Descent}
    \EndFor
\EndFor
\end{algorithmic}
\label{Algo:IterDistOptNCFun}
\end{algorithm}

\section{Convergence Analysis} \label{Sec:ConvAnalysis}
We state two important results that will be useful in convergence analysis, the first being on convergence of non-negative almost supermartingales by Robbins and Siegmund (Theorem 1, \cite{robbins1985convergence}) followed by Lemma 3.1 (b) by Ram \textit{et.al.}, \cite{ram2010distributed}.
\begin{lemma}  \label{Lem:RobSiegConv}
Let ($\Omega, \mathcal{F}, \mathcal{P}$) be a probability space and let $\mathcal{F}_0 \subset \mathcal{F}_1 \subset \ldots$ be a sequence of sub $\sigma-$fields of $\mathcal{F}$. Let $u_k, v_k$ and $w_k$, $k = 0, 1, 2, \ldots$ be non-negative $\mathcal{F}_k-$ measurable random variables and let \{$\gamma_k$\} be a deterministic sequence. Assume that $\sum_{k=0}^\infty \gamma_k < \infty$, and $\sum_{k=0}^\infty w_k < \infty$ and $$E[u_{k+1}|\mathcal{F}_k] \leq  (1+\gamma_k) u_k -v_k + w_k$$ holds with probability 1. Then, the sequence \{$u_k$\} converges to a non-negative random variable and $\sum_{k=0}^\infty v_k < \infty$.
\end{lemma}

\begin{lemma} \label{Lem:Ram}
Let $\{\zeta_k\}$ be a non-negative sequence scalar sequence. If $\sum_{k=0}^\infty \zeta_k < \infty$ and $0 < \beta < 1$, then $\sum_{k=0}^\infty \left( \sum_{j=0}^k \beta^{k-j} \zeta_j \right) < \infty$.
\end{lemma}

\noindent The well known non-expansive property (cf. \cite{bertsekas2003convex}) of Euclidean projection onto a non-empty, closed, convex set $\mathcal{X}$, is represented by the following inequality, $\forall \; x, y \in \mathbb{R}^D$, 
\begin{align}
    \| \mathcal{P}_\mathcal{X}[x] - \mathcal{P}_\mathcal{X}[y] \| \leq \|x - y\|. \label{Eq:EUCPROJINEQUAL}
\end{align}

We now present the relationship of server iterates between two consensus steps (at time instants 
$k$ and $k+1$) in the following Lemma.

\begin{lemma} \label{Lem:IterateConvRelation}
Under Assumptions~\ref{Asmp:Function}, \ref{Asmp:Set}, \ref{Asmp:SubBound}, \ref{Asmp:GradLip} and \ref{Asmp:DeltaConn}, the sequence of iterates generated by agents $x^J_k$ using Algorithm~\ref{Algo:IterDistOptNCFun}, satisfies for all $y \in \mathcal{X}$,
\begin{align}
    \sum_{J=1}^S \|x^J_{k+1} - y\|^2 &\leq \left(1 + 2\alpha_k \SB{N} \max_J \|\delta^J_k\| \right)\sum_{J=1}^S \|x^J_k - y\|^2  - 2 \alpha_k \left(f(\bar{x}_k) - f(y) \right) \nonumber \\
    & + 2 \alpha_k \left(\SB{L} + S\SB{N}\right) \max_J \|\delta^J_k\| + \alpha_k^2 \sum_{J=1}^S L_J^2. 
\end{align}
\end{lemma}
\begin{proof}
Every iteration in the algorithm involves two steps, a) information fusion using consensus step and b) projected gradient descent on local (possibly non-convex) objective function. 

\noindent The fused state ($v^J_k$) is obtained from the neighbor states using Eq.~\ref{Eq:InfFusALG1},
\begin{align}
    v^J_{k} = \sum_{I \in \mathcal{N}_J} B_{k}[J,I] x^I_{k}. 
\end{align}

\noindent The fused state is further updated based on projected gradient descent given by Eq.~\ref{Eq:ProjGradALG2},
\begin{align}
    x^J_{k+1} = \mathcal{P}_{\mathcal{X}}\left[v^J_{k} - \alpha_k g_J(v^J_{k}) \right]. \label{Eq:ProjGrad}
\end{align}

\noindent Using the non-expansive property of the projection operator used in Eq.~\ref{Eq:ProjGrad}, for all $y \in \mathcal{X}$ ($\mathcal{X}$ is a non-empty, closed convex set) gives us, 
\begin{align}
    \|x^J_{k+1} - y\|^2 &= \|\mathcal{P}_{\mathcal{X}}\left[v^J_{k} - \alpha_k g_J(v^J_{k}) \right] - y\|^2    && \ldots  \mathcal{P}_\mathcal{X}[y] = y \nonumber \\ 
    & \leq \|v^J_{k} - \alpha_k g_J(v^J_{k}) - y\|^2 \nonumber \\
    & \leq \|v^J_k - y\|^2 + \alpha_k^2 \|g_J(v^J_k)\|^2 - 2 \alpha_k (g_J(v^J_k))^T (v^J_k - y) \label{Eq:PGRelation1}
\end{align}

\noindent Adding the inequalities in Eq.~\ref{Eq:PGRelation1} for all agents $J = 1, 2, \ldots, S$ we get the following inequality,
\begin{align}
    \sum_{J=1}^S \|x^J_{k+1} - y\|^2 \leq \sum_{J=1}^S \|v^J_k - y\|^2 + \alpha_k^2 \sum_{J=1}^S \|g_J(v^J_k)\|^2 - 2 \alpha_k \sum_{J=1}^S (g_J(v^J_k))^T (v^J_k - y) 
\end{align}

\noindent We further use bounds on gradients  (Assumption~\ref{Asmp:SubBound}), $\|g_J(x)\| \leq L_J$ for all $J = 1, 2, \ldots, S$ and get,
\begin{align}
    \sum_{J=1}^S \|x^J_{k+1} - y\|^2 \leq \sum_{J=1}^S \|v^J_k - y\|^2 + \alpha_k^2 \sum_{J=1}^S L_J^2 - 2 \alpha_k \sum_{J=1}^S (g_J(v^J_k))^T (v^J_k - y) \label{Eq:IterateRelation1}
\end{align}

Now, we use consensus relationship used in information fusion. We start by stacking the state vector for all agents component wise. We use $\tilde{.}$ to denote a vector that is stacked by its coordinates (see Eq.~\ref{Eq:TildeVecDef} for definition). And, $\tilde{y}$ denotes $J$ copies of $y$ vector stacked coordinate wise. We know that in D-dimension the consensus step can be rewritten using Kronecker product of D-dimension identity matrix ($I_D$) and the doubly stochastic weight matrix ($B_k$) \cite{fax2004information}. 
\begin{align}
    \tilde{v}_{k} &= (I_D \otimes B_k) \tilde{x}_{k} && \ldots \text{Consensus Step} \label{Eq:E5}\\
    \tilde{v}_{k} - \tilde{y} &= (I_D \otimes B_k) (\tilde{x}_{k} - \tilde{y}) && \ldots \text{Eq.\eqref{Eq:E5} and } (I_D \otimes B_k)\tilde{y} = \tilde{y}
\end{align}

\noindent We now compare norms of both sides (2-norm),
\begin{align}
    \|\tilde{v}_{k} - \tilde{y}\|_2^2 &= \|(I_D \otimes B_k) (\tilde{x}_{k} - \tilde{y})\|_2^2 && \ldots \text{Norms of equal vectors are equal}  \\
    &\leq \|(I_D \otimes B_k)\|_2^2 \|(\tilde{x}_{k} - \tilde{y})\|_2^2 && \ldots \text{$\|Ax\|_2 \leq \|A\|_2 \|x\|_2$}
\end{align}

\noindent We use the following property of eigenvalues of Kronecker product of matrices. If $A$ ($m$ eigenvalues given by $\lambda_i$, with $i = 1, 2, \ldots, m$) and $B$ ($n$ eigenvalues given by $\mu_j$, with $j = 1, 2, \ldots, n$) are two matrices then the eigenvalues of the Kronecker product $A \otimes B$ are given by $\lambda_i \mu_j$ for all $i$ and $j$ ($mn$ eigenvalues). Hence, the eigenvalues of $I_D \otimes B_k$ are essentially $D$ copies of eigenvalues of $B_k$. Since $B_k$ is a doubly stochastic matrix, its eigenvalues are upper bounded by 1. Clearly, $\|(I_D \otimes B_k)\|_2^2 = \lambda_{\max}((I_D \otimes B_k)^\dagger(I_D \otimes B_k)) \leq 1$. This follows from the fact that $(I_D \otimes B_k)^\dagger(I_D \otimes B_k)$ is a doubly stochastic matrix since product of two doubly stochastic matrices is also doubly stochastic. 
\begin{align}
    \|\tilde{v}_{k} - \tilde{y}\|_2^2 &\leq \|(\tilde{x}_{k} - \tilde{y})\|_2^2.
\end{align} 

\noindent Furthermore, the square of the norm of a stacked vector is equal to sum of the square of the norms of all agents. 
\begin{align}
    \sum_{J=1}^S \|v^J_{k} - y\|^2 = \|\tilde{v}_{k} - \tilde{y}\|_2^2 \leq \|(\tilde{x}_{k} - \tilde{y})\|_2^2 = \sum_{J=1}^S \|x^J_{k} - y\|^2 \label{Eq:E6}
\end{align}

\noindent Merging the inequalities established above in Eq.~\ref{Eq:IterateRelation1} and Eq.~\ref{Eq:E6} we get, 
\begin{align}
    \sum_{J=1}^S \|x^J_{k+1} - y\|^2 \leq \sum_{J=1}^S \|x^J_k - y\|^2 + \alpha_k^2 \sum_{J=1}^S L_J^2 \underbrace{- 2 \alpha_k \sum_{J=1}^S (g_J(v^J_k))^T (v^J_k - y)}_{\Lambda}. \label{Eq:IterateRelation2}    
\end{align}

\noindent Typically, at this step one would use convexity of $f_J(x)$ to simplify the term $\Lambda$ in Eq.~\ref{Eq:IterateRelation2}. However, since $f_J(x)$ may be non-convex, and hence we need to follow a few more steps before we arrive at the iterate lemma.  

We further consider the fused state iterates $v^J_k$, the average $\bar{v}_k$ and the deviation of iterate from the average, $\delta^J_k = v^J_k - \bar{v}_k$. We further use gradient Lipschiztness (Assumption~\ref{Asmp:GradLip}) to arrive at the following relation,
\begin{align}
    g_J(v^J_k) = g_J(\bar{v}_k) + &l^J_k, \text{ where } \|l^J_k\| \leq N_J \|v^J_k - \bar{v}_k\| = N_J \|\delta^J_k\|. \label{Eq:UnrollGradient1} \\
    \max_{J} \|l^J_k\| &= \max_{J} \{ N_J \|\delta^J_k\| \} \leq \SB{N} \max_J \|\delta^J_k\| \label{Eq:UnrollGradient2Bound}
\end{align}

We use the above expressions in Eq.~\ref{Eq:IterateRelation2} to further bound the term $\Lambda$.
\begin{align}
    \Lambda &= - 2 \alpha_k \sum_{J=1}^S \left[ (g_J(v^J_k))^T (v^J_k - y) \right] = 2 \alpha_k \sum_{J=1}^S \left[ (g_J(v^J_k))^T (y - v^J_k) \right] \nonumber \\
    &= 2 \alpha_k \sum_{J=1}^S \left[ (g_J(\bar{v}_k) + l^J_k)^T (y - \bar{v}_k - \delta^J_k) \right] && \ldots v^J_k = \bar{v}_k + \delta^J_k \nonumber \\
    &= 2 \alpha_k \left[\underbrace{\sum_{J=1}^S g_J(\bar{v}_k)^T (y - \bar{v}_k)}_{T_1} +  \underbrace{\sum_{J=1}^S g_J(\bar{v}_k)^T(-\delta^J_k)}_{T_2}  + \underbrace{\sum_{J=1}^S (l^J_k)^T (y - v^J_k)}_{T_3}\right] \label{Eq:LambdaE1} 
\end{align}

\noindent Individual terms in Eq.~\ref{Eq:LambdaE1} can be bound in the following way,
\begin{align}
    T_1 &= \sum_{J=1}^S g_J(\bar{v}_k)^T (y - \bar{v}_k) = (\sum_{J=1}^S g_J(\bar{v}_k))^T(y - \bar{v}_k) && \ldots \text{$(y - \bar{v}_k)$ is independent of $J$} \nonumber \\
    &= g(\bar{v}_k)^T(y-\bar{v}_k) \leq f(y) - f(\bar{v}_k) && \ldots \text{$f(x)$ is convex} \\
    T_2 &= \sum_{J=1}^S g_J(\bar{v}_k)^T(-\delta^J_k) \leq \sum_{J=1}^S \|g_J(\bar{v}_k)^T\| \|(-\delta^J_k)\| \nonumber \\ 
    &\leq \max_J \|\delta^J_k\| \sum_{J=1}^S L_J \leq \SB{L} \max_J \|\delta^J_k\| && \ldots \text{$\|g_J(x)\|\leq L_J$ and $\SB{L} = \sum_{J=1}^S L_J$} \\
    T_3 &= \sum_{J=1}^S (l^J_k)^T (y - v^J_k) \leq \max_J \|l^J_k\| \sum_{J=1}^S \|v^J_k - y\| \nonumber \\
    &\leq  \SB{N} \max_J \|\delta^J_k\| \sum_{J=1}^S \|v^J_k - y\| && \ldots \text{Eqs.~\ref{Eq:UnrollGradient1}, \ref{Eq:UnrollGradient2Bound} and $\SB{N} = \sum_{J=1}^S N_J$}
\end{align}

\noindent We can further use the property $2 \|a\| \leq 1 + \|a\|^2$ to bound term $T_3$.
\begin{align}
    T_3 &\leq \SB{N} \max_J \|\delta^J_k\| \sum_{J=1}^S \|v^J_k - y\| \leq \SB{N} \max_J \|\delta^J_k\| \left[ \sum_{J=1}^S \left( 1+ \|v^J_k - y\|^2\right) \right] \nonumber \\
    & \leq \SB{N} \max_J \|\delta^J_k\| \left[S + \sum_{J=1}^S \|v^J_k - y\|^2 \right] \leq \SB{N} \max_J \|\delta^J_k\| \left[S + \sum_{J=1}^S \|x^J_k - y\|^2 \right] && \ldots \text{Eq.~\ref{Eq:E6}}
\end{align}

\noindent We can use the bounds on $T_1, T_2$ and $T_3$ to get a bound on $\Lambda$.
\begin{align}
    \Lambda \leq 2 \alpha_k \left( -\left(f(\bar{v}_k) - f(y) \right) + \SB{L} \max_J \|\delta^J_k\| + \SB{N} \max_J \|\delta^J_k\| \left[S + \sum_{J=1}^S \|x^J_k - y\|^2 \right]  \right)
\end{align}
\noindent Note that we can replaced, $f(\bar{v}_k)$ with $f(\bar{x}_k)$. This follows from the fact that doubly stochastic matrices preserve iterate averages, i.e. $\bar{v}_k = \bar{x}_k$ (cf. \cite{nedic2009distributed}).
\noindent The iterate update relation hence becomes, 
\begin{align}\noindent 
    \sum_{J=1}^S \|x^J_{k+1} - y\|^2 &\leq \left(1 + 2\alpha_k \SB{N} \max_J \|\delta^J_k\| \right)\sum_{J=1}^S \|x^J_k - y\|^2  - 2 \alpha_k \left(f(\bar{x}_k) - f(y) \right) \nonumber \\
    & + 2 \alpha_k \left(\SB{L} + S\SB{N}\right) \max_J \|\delta^J_k\| + \alpha_k^2 \sum_{J=1}^S L_J^2. 
\end{align}
$\hfill \blacksquare$
\end{proof}


\begin{lemma} \label{Lem:deltaJbound}
Let iterates be generated by Algorithm~\ref{Algo:IterDistOptNCFun}, while Assumptions~\ref{Asmp:Function}, \ref{Asmp:Set}, \ref{Asmp:SubBound}, \ref{Asmp:GradLip}, and  \ref{Asmp:DeltaConn} hold, then there exists constant $\nu<1$, such that the following bound on the maximum (over $J$) disagreement between iterate at agent $J$ and the average iterate given by $\delta^J_k$ (Eqs.~\ref{Eq:deltaDef} and \ref{Eq:deltaDef2}) holds,
$$\max_J\{\|\delta^J_{k+1}\|\} \leq \frac{S-1}{S} \left(\nu^{k+1} \max_{P,Q} \left( \|x^P_{0,0} - x^Q_{0,0}\| \right) + \SB{L} \sum_{i = 1}^k \left( \alpha_i \nu^{k-i} \right) \right).$$
\end{lemma}

\begin{proof}
We use Kronecker product to write consensus step as shown in Eq.~\ref{Eq:E5}. This step is equivalent to the following form of representing the consensus step,
\begin{equation}
v^I_{k} = \sum_{J=1}^S B_k [I,J] x^J_{k} \label{Eq:InfoFused1}
\end{equation}
where $v^I_{k}$ and $x^J_{k}$ represent the fused parameter vector at agent $I$ and parameter vector at agent $J$ at time step $\{k\}$ while $B_k [I,J]$ is a scalar representing the $I^{th}$ row and $J^{th}$ column entries of matrix $B_k$. We know from Eq.~\ref{Eq:InfoFused1} that the difference between fused parameter vector at agent $I$ and $J$ can be written as,
\begin{equation}
v^J_k - v^I_k = \sum_{L=1}^S \left(B_k[J,L] - B_k[I,L]\right) x^L_k.
\label{Eq:ROWSTOC}
\end{equation}

Since, $B_k$ is doubly stochastic, clearly the coefficients of states in Eq.~\ref{Eq:ROWSTOC} add up to zero (i.e. $\sum_{J=1}^S \left( B_k[J,L] - B_k[I,L] \right) = 0$). Collecting all positive coefficients and negative coefficients and rearranging we get the following equation,
\begin{align}
v^J_{k+1} - v^I_{k+1} = \sum_{P,Q} \eta_{P,Q} (x^P_{k} - x^Q_{k}), \qquad \ldots \forall \; I, G 
\label{Eq:REARRANG}
\end{align}
where, $\eta_{P,Q} \geq 0$ is the weight associated to servers $P$ and $Q$ and $\eta_{P,Q} \geq 0$. Note that all coefficients $\eta_{P,Q}$ refer to some $J, I$ pair at time $k$. For simplicity in notation we will ignore $I, J$ and $k$ without any loss of generality or correctness.

\begin{assumption} \label{Asmp:Scrambling}
We assume the transition matrix $B_k$ to be a scrambling matrix. \cite{seneta2006non}
\end{assumption}
We have from Assumption~\ref{Asmp:Scrambling}, any two rows of  $B_k$ matrix have a non-zero column entry. And since any entry of the matrix is less than 1, the difference is also strictly less than 1. Hence, $\sum \eta_{PQ} < 1$. By taking norm on both sides of Eq.~\ref{Eq:REARRANG}, recalling  Assumption~\ref{Asmp:Scrambling} and using triangle inequality we get, for all $I,J$, 
\begin{equation}
\|v^J_{k} - x^I_{k}\| \leq \left( \sum \eta_{P,Q} \right) \max_{P,Q} \|x^P_{k} - x^Q_{k}\|. 
\end{equation}
Since the above inequality is valid for all $I, G$, we can rewrite the above relation as,
\begin{align}
 \max_{I, J} \|v^J_{k} - v^I_{k}\| \leq \max_{P,Q}\left( \sum \eta_{P,Q} \right) \max_{P,Q} \|x^P_{k} - x^Q_{k}\|. \label{Eq:LAB1}
\end{align}
Note that, $\max_{P,Q} \left( \sum \eta_{P,Q} \right)$ is dependent only on the topology at time $k$ (i.e. the doubly stochastic weight matrix given by $B_k$). Due to the countable nature of possible topologies for $S$ agents, we can define a new quantity $\nu = \max_k \{ \max_{P,Q} \{ \sum \eta_{P,Q} \} \}$ \footnote{Note that $\eta_{P,Q}$ is dependent on $I, J$ pair and $k$.}. By definition, $\max_{P,Q} \{ \sum \eta_{P,Q} \} \leq \nu, \; \forall \; k \geq 0$ and since $\max_{P,Q} \{ \sum \eta_{P,Q} \} < 1 \; \forall \; k \geq 0$, we have $\nu < 1$.

We can write the difference between parameter vectors at agent $I$ and $J$ as,
\begin{align}
x^J_{k+1} - x^I_{k+1} = \mathcal{P}_\mathcal{X} \left[ v^J_k - \alpha_k g_J(v^J_k) \right] - \mathcal{P}_\mathcal{X} \left[ v^I_k - \alpha_k g_I(v^I_k) \right] && \ldots \text{Eq.~\ref{Eq:ProjGradALG2}}
\end{align}
and, further obtain inequality bound using non-expansive property of the projection operator ($\mathcal{X}$ is a non-empty, closed-convex set),
\begin{align}
\|x^J_{k+1} - x^I_{k+1}\|^2 &\leq \| v^J_k - \alpha_k g_J(v^J_k) - v^I_k + \alpha_k g_I(v^I_k) \|^2 \nonumber \\
&\leq \|v^J_k - v^I_k\|^2 + \alpha_k^2 \|g_J(v^J_k) - g^I(v^I_k)\|^2 + 2 \alpha_k \|v^J_k - v^I_k\| \|g^J(x^J_k) - g^I(x^I_k)\|  \nonumber \\
&\leq \|v^J_k - v^I_k\|^2 + \alpha_k^2 (L_J + L_I)^2 + 2 \alpha_k (L_J+L_I) \|v^J_k - v^I_k\| \nonumber \\
&\leq (\|v^J_k - v^I_k\| + \alpha_k (L_J + L_I))^2 \leq (\|v^J_k - v^I_k\| + \alpha_k \SB{L})^2 \nonumber \\
\|x^J_{k+1} - x^I_{k+1}\| &\leq \|v^J_k - v^I_k\| + \alpha_k \SB{L} \label{Eq:DeltaJBound1}
\end{align}

Now we perform maximization on both sides of Eq.~\ref{Eq:DeltaJBound1} to get,
\begin{align}
\max_{I,J} \|x^J_{k+1} - x^I_{k+1}\| \leq \max_{I,J} \|v^J_k - v^I_k\| + \alpha_k \SB{L},
\end{align}
and further use the bound in Eq.~\ref{Eq:LAB1},
\begin{align}
\max_{I,J} \|x^J_{k+1} - x^I_{k+1}\| \leq \nu \max_{P,Q} \|x^P_{k} - x^Q_{k}\| + \alpha_k \SB{L}.
\end{align}

\noindent Now we perform an unrolling operation, and relate the maximum agent disagreement to the initial disagreement between agents (at step $0$).
\begin{align}
\max_{I , J} \|x^J_{k+1} - x^I_{k+1}\| &\leq  \left( \nu \max_{P,Q} \left( \|x^P_{k} - x^Q_{k}\| \right) + \alpha_k \SB{L} \right) \nonumber \\
&\leq \left( \nu \left( \nu \max_{P,Q} \left( \|x^P_{k-1} - x^Q_{k-1}\| \right) + \alpha_{k-1} \SB{L} \right) + \alpha_k  \SB{L} \right) \nonumber \\
&\leq \qquad \ldots \nonumber \\ 
&\leq \left(\nu^{k+1} \max_{P,Q} \left( \|x^P_{0} - x^Q_{0}\| \right) + \SB{L} \sum_{i = 1}^k \left( \alpha_i \nu^{k-i} \right) \right) \label{Eq:ASYMPCONVPARVEC}
\end{align}

\noindent We start with the definition of $\delta^J_{k+1}$ (see  Eq.~\ref{Eq:deltaDef2}) and consider the maximum over all agents,
\begin{align}
\max_J\{\|\delta^J_{k+1}\|\} &= \max_J\{ \|x^J_{k+1} - \bar{x}_{k+1}\| \} = \max_J\{ \|x^J_{k+1} - \frac{1}{S} \sum_{I=1}^S x^I_{k+1}\| \} \nonumber \\
&= \max_J\{ \|\frac{1}{S} \sum_{I \neq J}^S (x^J_{{k+1}} - x^I_{{k+1}})\| \} \leq \frac{S-1}{S} \max_{I,J} \|x^J_{{k+1}} - x^I_{{k+1}}\|. \label{Eq:ScramRes1}
\end{align}
Together with Eq.~\ref{Eq:ASYMPCONVPARVEC}, we arrive at the desired expression from the statement of lemma, 
\begin{align}
\max_J\{\|\delta^J_{k+1}\|\} \leq \frac{S-1}{S} \left(\nu^{k+1} \max_{P,Q} \left( \|x^P_{0} - x^Q_{0}\| \right) + \SB{L} \sum_{i = 1}^k \left( \alpha_i \nu^{k-i} \right) \right) \label{Eq:MAXDJBOUND}
\end{align} $\hfill \blacksquare $
\end{proof}

Note that we can do away with Assumption~\ref{Asmp:Scrambling}, and prove similar bound on the maximum disagreement between agent iterates and its average for any connected graph. For simplicity, we assume that the transition matrix is scrambling. 

\begin{claim} [Consensus] \label{Cl:Consensus}
The agent parameter vectors achieve consensus asymptotically. $$\lim_{k \rightarrow \infty} \max_{I,J} \|x^J_{k+1} - x^I_{k+1}\| = 0.$$
\end{claim}
\begin{proof}
We know from Eq.~\ref{Eq:ASYMPCONVPARVEC} that the maximum disagreement between any two agents ($I$ and $J$) at time $k$ is given by,
\begin{align}
    \max_{I,J} \|x^J_{k+1} - x^I_{k+1}\| \leq \left( \nu^{k+1} \max_{P,Q} \left( \|x^P_{0} - x^Q_{0}\| \right) + \SB{L} \sum_{i = 1}^k \left( \alpha_i \nu^{k-i} \right) \right). \label{Eq:SerConTemp1} 
\end{align}

The first term on the right hand side of above expression tends to zero as $k \rightarrow \infty$, since $\nu < 1$ and $\nu^{k+1} \rightarrow 0$ as $k \rightarrow 0$. 

Let us consider $\epsilon_0 > 0$, and define $0 < \epsilon < \epsilon_0 \frac{1 - \nu}{2\SB{L}\nu}$. Since, $\epsilon_0 > 0$ and $\nu < 1$ we know that such an $\epsilon$ exists. We now show that the second term in Eq.~\ref{Eq:SerConTemp1}, decreases to zero too. Since $\alpha_k$ is non-increasing sequence, $\exists \ K = K(\epsilon) \in \mathbb{N}$ such that $\alpha_i < \epsilon$ for all $i \geq K$. Hence we can rewrite the second term for $k > K$ as,
$$ \SB{L} \sum_{i = 1}^k  \left( \alpha_i \nu^{k-i} \right) = \SB{L} \left[\underbrace{\left( \alpha_0 \nu^k + \alpha_1 \nu^k-1 + \ldots + \alpha_{K-1} \nu^{k-K+1}\right)}_{A} + \underbrace{ \left(\alpha_{K} \nu^{k-K} + \ldots + \alpha_k \nu^0 \right)}_{B}\right]$$
\noindent We can bound the individual terms A and B by using the monotonically non-increasing property of $\alpha_i$ and sum of a geometric series.
\begin{align}
A &=  \alpha_0 \nu^k + \alpha_1 \nu^{k-1} + \ldots + \alpha_{K-1} \nu^{k-K+1} \nonumber \\
&\leq \alpha_0 (\nu^k + \nu^{k-1} + \ldots + \nu^{k-K+1}) && \ldots \alpha_1 \geq \alpha_i \; \forall \ i \geq 1 \nonumber \\
&\leq \alpha_0 \nu^{k-K+1} \left(\frac{1 - \nu^{K}}{1-\nu} \right) \leq \frac{\alpha_0 \nu^{k-K+1}}{1-\nu}   && \ldots \nu < 1 \implies 1 - \nu^{K} < 1, \; \forall \ k > K\label{Eq:ABOUNDS1}\\
B &= \alpha_K \nu^{k-K} + \ldots + \alpha_k \nu^0 \nonumber \\
& < \epsilon \nu \left( \frac{1-\nu^{k-K+1}}{1-\nu}\right) \leq \frac{\epsilon \nu}{1-\nu}  && \ldots \alpha_i < \epsilon, \; \forall i \geq K \text{ and } \nu < 1 \label{Eq:BBOUNDS1}
\end{align}

\noindent Since the right side of inequality in Eq.~\ref{Eq:ABOUNDS1} is monotonically decreasing in $k$ ($\nu < 1$) with limit $0$ as $k \rightarrow \infty$. Hence $\exists K_{0_1} > K$ such that $\nu^{k-K+1} < \epsilon$, $\forall \; k \geq K_{0_1}$ and hence $A < \frac{\epsilon_0}{2 \SB{L}}$. Substituting the upper bound for $\epsilon$ in right side of inequality in Eq.~\ref{Eq:BBOUNDS1}, we get $\exists K_{0_2} > K$ such that $B < \frac{\epsilon_0}{2 \SB{L}}$, $\forall \; k \geq K_0$. 

Using the bounds obtained above (on $A$ and $B$), we conclude, $\forall \; \epsilon_0 > 0$, $\exists \; K_0 = \max \{K_{0_1}, K_{0_2}\}$ such that $\SB{L} \sum_{i=0}^k \left( \alpha_i \nu^{k-i} \right) < \epsilon_0$, $\forall \; k > K_0$. Clearly (from the $\epsilon-\delta$ definition of limit), $\lim_{k \rightarrow \infty} \SB{L} \sum_{i=0}^k \left( \alpha_i \nu^{k-i} \right) = 0$. This limit together with the limit of first term on the right side of Eq.~\ref{Eq:SerConTemp1} being zero, implies, $\lim_{k \rightarrow \infty} \max_{I,J} \|x^J_{k+1} - x^I_{k+1}\| = 0$. Thus we have asymptotic consensus of the agent parameter vectors.

$\hfill \blacksquare$
\end{proof}

\begin{theorem} \label{Th:ConvMain}
Let Assumptions~\ref{Asmp:Function}, \ref{Asmp:Set}, \ref{Asmp:SubBound}, \ref{Asmp:GradLip} and \ref{Asmp:DeltaConn} hold with $\mathcal{X}^*$ being a nonempty bounded set. Also assume a diminishing step size rule presented in Eq.~\ref{Eq:LearnStepCond}. Then, for a sequence of iterates $\{x^J_{k}\}$ generated by an distributed optimization algorithm (Algorithm~\ref{Algo:IterDistOptNCFun}) the  iterate average ($\bar{x}_k$) converge to an optimum in $\mathcal{X}^*$.
\end{theorem}

\begin{proof}
We intend to prove convergence using deterministic version of Lemma~\ref{Lem:RobSiegConv}. We begin by using the relation between iterates given in Lemma~\ref{Lem:IterateConvRelation} with $y = x^* \in \mathcal{X}^*$,
\begin{align}
    \sum_{J=1}^S \|x^J_{k+1} - x^*\|^2 &\leq \left(1 + \underbrace{2\alpha_k \SB{N} \max_J \|\delta^J_k\|}_{\gamma_k} \right)\sum_{J=1}^S \|x^J_k - x^*\|^2  - 2 \alpha_k \left(f(\bar{x}_k) - f(x^*) \right) \nonumber \\
    & + \underbrace{2 \alpha_k \left(\SB{L} + S\SB{N}\right) \max_J \|\delta^J_k\| + \alpha_k^2 \sum_{J=1}^S L_J^2}_{w_k}. \label{Eq:PROOF0}
\end{align}

We check if the above inequality satisfies the conditions in Lemma~\ref{Lem:RobSiegConv} viz. $\sum_{k=0}^\infty \gamma_k < \infty$ and $\sum_{k=0}^\infty w_k < \infty$. $\gamma_k$ and $w_k$ are defined as shown in Eq.~\ref{Eq:PROOF0}.

We first show that $\sum_{k = 0}^\infty \alpha_k \max_J \|\delta^J_k\| < \infty$.
\begin{align}
    \sum_{k = 0}^\infty & \alpha_k \max_J \|\delta^J_k\| \leq \frac{S-1}{S} \sum_{k = 0}^\infty \alpha_k  \left( \nu^{k+1} \max_{P,Q} \left( \|x^P_{0} - x^Q_{0}\| \right) + \SB{L} \sum_{i = 1}^k \left( \alpha_i \nu^{k-i} \right) \right) && \ldots \text{Eq.~\ref{Eq:MAXDJBOUND}} \nonumber \\
    &\leq \frac{S-1}{S} \left( \sum_{k = 0}^\infty \alpha_k  \nu^{k+1} \max_{P,Q} \left( \|x^P_{0} - x^Q_{0}\| \right) + \sum_{k = 0}^\infty \SB{L} \alpha_k \sum_{i = 1}^k \left( \alpha_i \nu^{k-i} \right) \right) \nonumber \\
    &\leq \frac{S-1}{S} \left( \max_{P,Q} \left( \|x^P_{0} - x^Q_{0}\| \right) \sum_{k = 0}^\infty \alpha_k  \nu^{k+1}  + \sum_{k = 0}^\infty \SB{L} \sum_{i = 1}^k \left( \alpha_i^2 \nu^{k-i} \right) \right) && \ldots \alpha_k \leq \alpha_i, \forall i \leq k \nonumber \\
    &\leq \frac{S-1}{S} \left( \max_{P,Q} \left( \|x^P_{0} - x^Q_{0}\| \right) \sum_{k = 0}^\infty \alpha_k  \nu^{k+1}  + \SB{L} \sum_{k = 0}^\infty \sum_{i = 1}^k \left( \alpha_i^2 \nu^{k-i} \right) \right). \label{Eq:PROOF1}
\end{align}
In the above expression, we can show that the first term is convergent by using the ratio test. We observe that, $$ \limsup_{k \rightarrow \infty} \frac{\alpha_{k+1} \nu^{k+2}}{\alpha_{k} \nu^{k+1}} = \limsup_{k \rightarrow \infty} \frac{\alpha_{k+1} \nu}{\alpha_k} < 1 \implies \sum_{k=0}^\infty \alpha_k \nu^k < \infty,$$
since, $\alpha_{k+1} \leq \alpha_k$ and $\nu < 1$. Arriving at the second term involves using the non-increasing property of $\alpha_i$, i.e. $\alpha_k \leq \alpha_i \forall i \leq k $. Now, we use Lemma~\ref{Lem:Ram}, with $\zeta_j = \alpha_j^2$ (where $\sum_{k=0}^ \infty \zeta_k < \infty$) and show that the second term in the above expression is finite, i.e. $(S-1)/S \SB{L} \sum_{k = 0}^\infty \sum_{i = 1}^k \left( \alpha_i^2 \nu^{k-i} \right) < \infty$. Together using finiteness of both parts on the right side of Eq.~\ref{Eq:PROOF1} we have proved, 
\begin{align}
\sum_{k = 0}^\infty \alpha_k \max_J \|\delta^J_k\| < \infty
\label{Eq:PROOF2}
\end{align}

We now begin to prove the finiteness of sum of $\gamma_k$ sequence, $\sum_{k=0}^\infty \gamma_k < \infty$. 
\begin{align}
\sum_{k=0}^\infty \gamma_k &= 2 \SB{N} \sum_{k=0}^\infty \left(\alpha_k \max_J\|\delta^J_k\| \right) < \infty.
\end{align}

We can similarly prove $\sum_{k=0}^\infty w_k < \infty$.
\begin{align}
    \sum_{k=0}^\infty w_k  = 2 (\SB{L}+S\SB{N}) \sum_{k=0}^\infty \alpha_k \max_J \|\delta^J_k\| +  (\sum_{J=1}^S L_J^2)\sum_{k=0}^\infty \alpha_k^2 < \infty
\end{align}
The first term above is finite as proved in Eq.~\ref{Eq:PROOF2} and the second term is finite due due to the assumption on learning rate (Eq.~\ref{Eq:LearnStepCond}).

We can now use the deterministic version of Lemma~\ref{Lem:RobSiegConv} to show the convergence of iterate average to the optimum. We know from proof above that $\sum_{k=0}^\infty \gamma_k < \infty$ and $\sum_{k=0}^\infty w_k < \infty$.  As a consequence of Lemma~\ref{Lem:RobSiegConv}, we get that the sequence $\eta_k^2$ converges to some point and $\sum_{k=0}^\infty 2 \frac{\alpha_k M}{S} (f(\bar{x}_{k}) - f(x^*)) < \infty$.

We use $\sum_{k=0}^\infty 2 \frac{\alpha_k M}{S} (f(\bar{x}_{k}) - f(x^*)) < \infty$ to show the convergence of the iterate-average to the optimum. Since we know $\sum_{k=0}^\infty \alpha_k = \infty$, it follows directly that $\lim \inf_{k \rightarrow \infty} f(\bar{x}_{k}) = f(x^*)$. And due to the continuity of $f(x)$, we know that the sequence of iterate average must enter the optimal set $\mathcal{X}^*$ (i.e. $\bar{x}_{k} \in \mathcal{X}^*$). Since $\mathcal{X}^*$ is bounded (compactness in $\mathbb{R}^D$), we know that there exists a iterate-average subsequence $\bar{x}_{k_l} \subseteq \bar{x}_k$ that converges to some $x^* \in \mathcal{X}^*$. 

We know from Claim~\ref{Cl:Consensus} that the agents agree to a parameter vector asymptotically (i.e. $x^J_{k} \rightarrow x^I_{k}, \ \forall I \neq J$ as $k \rightarrow \infty$). Hence, all agents agree to the iterate average. This along with the convergence of iterate-average to the optimal solution gives us that all agents converge to the optimal set $\mathcal{X}^*$ (i.e. $x^J_{k} \in \mathcal{X}^*, \ \forall J, \text{ as } k \rightarrow \infty$). $\hfill \blacksquare$
\end{proof}

\subsection{Extension} \label{Sec:Extensions}
A graph is called $Q$-connected, if agents form a connected component at least once every $Q$ iterations. We can relax the requirement on connectedness and easily make similar claims for a $Q$ connected graph (instead of Assumption~\ref{Asmp:DeltaConn}). Using the analysis technique developed above (also see \cite{gade16distoptclientserver}) and the analysis in \cite{ram2010distributed}, it is straightforward to show that Algorithm~\ref{Algo:IterDistOptNCFun} can optimize convex sum of non-convex functions as posed in Section~\ref{Sec:Problem} for $Q$-connected topology. 

\section{Discussion}
In this report we show that distributed optimization algorithms (Algorithm~\ref{Algo:IterDistOptNCFun}) can correctly optimize a convex function with non-convex partitions. The analysis technique developed above easily allows for other extensions as mentioned above in Section~\ref{Sec:Extensions}.

\subsection{Privacy} \label{Sec:Privacy}
Privacy has emerged to become one of the most important and challenging aspect of machine learning and distributed optimization. We propose two methods that can enhance privacy in distributed optimization. Both methods can be easily shown to perform distributed optimization correctly using results and analysis techniques proposed in this report.  

\subsubsection*{Function Partitioning}
The first approach to introduce privacy in distributed optimization is by constructing fictitious partitions of the individual objective function ($f_i(x)$). Further, these partitions are used in gradient descent step and several different state updates are created by an agent. Now these dissimilar states are shared with different neighbors. We show in Section~\ref{Sec:Motivation} that this strategy can be easily analyzed and proved to work correctly by using the analysis developed in this work. We hypothesize that selecting (and/or constructing) these function partitions dynamically can be a very successful strategy of introducing privacy in distributed optimization. Details about the strategy and privacy analysis for this strategy will be explored in a future technical report. 

\subsubsection*{Random Function Sharing}
We propose an alternate privacy enhancing strategy inspired from secure multi-party aggregation algorithm in \cite{abbe2012privacy}. In this strategy, every agent $I$ sends a randomly generated deterministic function $R_{I,J}(x)$ to neighboring agents $J$. These transmissions are assumed to be secure. Hence, any agent $I_0$ has access to all randomly generated deterministic functions that it has transmitted ($R_{I_0,K}(x)$, for some $K$) and those that it has received ($R_{P,I_0}(x)$, for some $P$). The distributed optimization problem retains its structure with $f_i(x)$ being replaced by $\hat{f}_i(x)$, 
$$\hat{f}_i(x) \triangleq f_i(x) + \sum_{P \ : \ i \in\mathcal{N}_P} R_{P,i}(x) - \sum_{K \in \mathcal{N}_i} R_{i,K}(x).$$
Note, since the randomly generated functions may not necessarily be convex, our new individual objective functions need not be convex. 

It is not hard to see that the sum of all new individual functions ($\hat{f}_i(x)$) is equal to the sum of all old individual functions ($f_i(x)$). This follows from the fact that the randomly generated functions cancel out during aggregation of individual objective functions.
$$\sum_{i=1}^S \hat{f}_i(x) = \sum_{i=1}^S \left( f_i(x) + \sum_{P \ : \ i \in\mathcal{N}_P} R_{P,i}(x) - \sum_{K \in \mathcal{N}_i} R_{i,K}(x) \right) = \sum_{i=1}^S {f}_i(x)$$

Directly applying convergence results from this report, we can state that the distributed protocol in Algorithm~\ref{Algo:IterDistOptNCFun} with new individual objective function will solve the original problem (minimizing $f(x) = \sum_{i=1}^S f_i(x)$). Details about the strategy and privacy analysis for this strategy will be explored in a future technical report.

\bibliography{Central.bib}
\bibliographystyle{ieeetr}

\end{document}